\documentclass[12pt]{article}
\usepackage{e-jc}
\specs{P4.3}{21(4)}{2014}

\usepackage{amsthm,amsmath,amssymb,float}
\usepackage{verbatim} 
\usepackage[pdftex]{graphicx}
\usepackage{url}
\usepackage[english]{babel}

\usepackage{needspace}
\usepackage{color}

\definecolor{clgray}{rgb}{0.25,0.25,0.25}

\theoremstyle{plain}
\newtheorem{theorem}{Theorem}[section]

\newtheorem{lemma}[theorem]{Lemma}

\theoremstyle{definition}
\newtheorem{definition}[theorem]{Definition}

\theoremstyle{remark}

\newcommand{\specialcell}[2][c]{%
  \begin{tabular}[#1]{@{}c@{}}#2\end{tabular}}

\definecolor{CLNOTE}{rgb}{0,0,.6}
\definecolor{CLNOTERED}{rgb}{.8,0,0}
\definecolor{CLTYPO}{rgb}{.8,0,.8}

\title{\bf 
Permutation Reconstruction from Differences}

\author{Marzio De Biasi\\
\small\tt marziodebiasi@gmail.com}

\date{{\small The Electronic Journal of Combinatorics {\bf 21(3)} (2014), \#P3.36}}

\begin{document}

\maketitle

\begin{abstract}
We prove that the problem of reconstructing a permutation
$\pi_1,\dotsc,\pi_n$ of the integers $[1\dotso n]$ given the absolute
differences $|\pi_{i+1}-\pi_i|$, $i = 1,\dotsc,n-1$ is $\sf{NP}$--complete.
As an intermediate step we first prove the $\sf{NP}$--completeness of
the decision version of a new puzzle game that 
we call \emph{ Crazy Frog Puzzle}. 
The permutation reconstruction from differences is one of the simplest
combinatorial problems that have been proved to be computationally intractable.
\end{abstract}

\section{Introduction}

Permutation reconstruction has been studied as a variation of the
graph reconstruction problem that arose from an unsolved conjecture of
Ulam \cite{ulam60}. Given an unlabeled graph $A$, deleting
one of its vertex together with its incident edges in each possible way we obtain
the minors $A_1,...,A_n$. Ulam's conjecture says that given two graphs $A,B$
with $n > 3$ vertices, if there exists a bijection $\alpha:\{1,...,n\}\to\{1,...,n\}$
such that $A_i$ is isomorphic to $B_{\alpha(i)}$ then $A$ is isomorphic to $B$.
In the permutation reconstruction variant, we consider a permutation $p$ with
$n$ entries; we can delete $k$ of the entries in each possible way and renumber them with respect to order and obtain $\binom{n}{k}$ permutations of the numbers from $1$ to $n-k$ that are called $(n-k)$-minors.
Smith \cite{DBLP:journals/combinatorics/Smith06}
introduced the problem of reconstructing the original permutation $p$ from
the multiset $M_k(p)$ of its $(n-k)$-minors and looked at the number $N_k$ defined to
be the smallest number such that we can reconstruct permutations of length
$n \geq N_k$. Raykova \cite{DBLP:journals/combinatorics/Raykova06}
proved the existence of $N_k$ for every positive integer $k$
and gave an upper bound of $N_k < \frac{k^2}{4} + 2k + 4$ and a lower bound
$N_k > k + \log_2{k}$. Monks \cite{DBLP:journals/combinatorics/Monks09} studied
the reconstruction of a permutation $p$ from its set of cycle minors, 
where each $i$-th cycle minor is obtained deleting the entry $i$ from the
decomposition of $p$ into disjoint cycles and reducing each remaining
entry larger than $i$ by $1$.
He showed that any permutation of $\{1,2,...,n\}$ can be reconstructed from
its set of cycle minors if and only if $n \geq 6$.

In this paper we focus our attention to another variant that has an even
simpler formulation:
the problem of reconstructing a permutation $\pi_1,\dotsc,\pi_n$
of the integers $[1\dotso n]$ given the absolute
differences $|\pi_{i+1}-\pi_i|$, $i = 1,\dotsc,n-1$.
Up to our knowledge, the first formulation of the problem was proposed by
Mohammad Al-Turkistany in a question posted on MathOverflow.net,
a question and answer site for professional mathematicians \cite{blog:permrecons}.
We will prove that deciding if such a permutation exists is $\sf{NP}$--complete:
any given solution can be quickly verified in polynomial time, but
there is no known way to build an efficient polynomial time algorithm
that find a solution in the first place. And such an efficient algorithm
doesn't exist unless $\sf{P} = \sf{NP}$, which is the major open problem
in computer science.

In order to prove our result we first introduce
a new puzzle game, the \emph{Crazy Frog Puzzle},  with
the following rules: a frog is placed on a square
grid board; some cells of the grid are blocked, some are empty.
The frog must follow a given sequence of horizontal, vertical or
diagonal jumps of varying length; at every jump the frog can only decide to
follow the given direction or jump in the opposite direction.
For example, when facing an horizontal jump of length two, the frog
placed on cell $(x,y)$ can jump left on cell $(x-2,y)$ or right on cell
$(x+2,y)$.
The frog cannot jump outside
the board, on a blocked cell, or on a cell that has already been visited.
The aim of the game is to choose the correct directions of
the jumps and make the frog visit all the empty cells of the board
exactly once.
Figure~\ref{fig:crazyfrog} shows an instance of the Crazy Frog Puzzle and its solution.

\begin{figure}[htp]
\centering
\includegraphics[width=10cm]{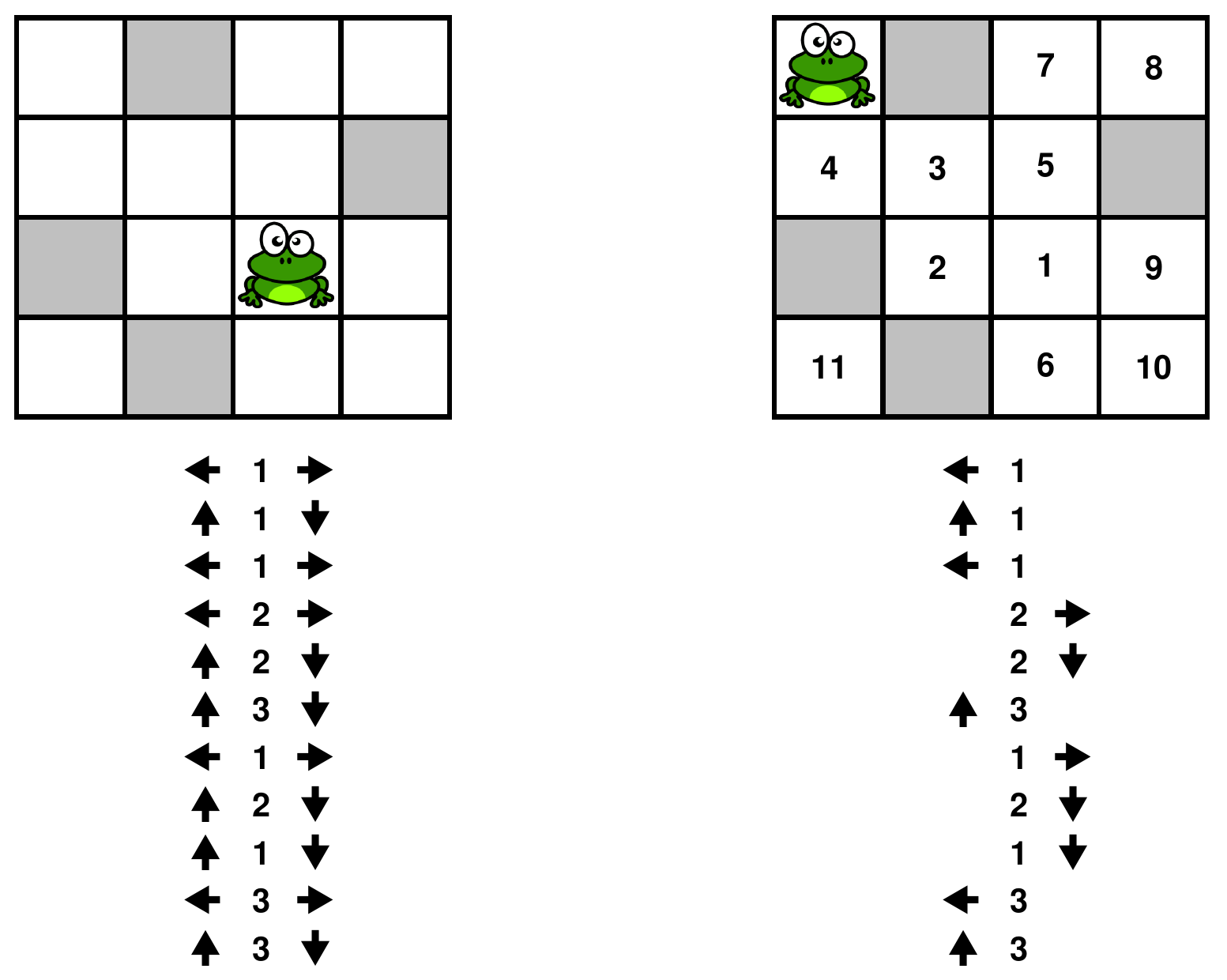}
\caption{An instance of the Crazy Frog Puzzle on the left and its solution on the right.}\label{fig:crazyfrog}
\end{figure}

In line with the recent interest in the complexity of puzzle games
\cite{DBLP:journals/icga/KendallPS08}
\cite{AlgGameTheory_GONC3},
we study how hard it can be to solve a Crazy Frog Puzzle game. 
In Section~\ref{sec:cfp} we formally define the \emph{decision} version
of the Crazy Frog Puzzle (CFP) and give some terminology; in Section~\ref{sec:npc}
we prove that deciding if a given Crazy Frog Puzzle has a solution
is $\sf{NP}$--complete. In Section~\ref{sec:1d} we prove
that CFP remains $\sf{NP}$-complete even
if the board is 1-dimensional. Finally in Section~\ref{sec:permutation}
we prove that the problem remains $\sf{NP}$--complete even if
the initial board has no blocked cells and we will show that
the 1-dimensional CFP without blocked cells is equivalent
to the permutation reconstruction from differences problem.

\section{Crazy Frog Puzzle}
\label{sec:cfp}

We define the decision problem \cite{sipser} that asks whether a given Crazy Frog Puzzle 
has a solution or not:

\begin{definition}[\sc{Crazy Frog Puzzle}]~\\~\\
{\bf Input}: An $n \times n$ partially filled board: some cells are {\em empty} some cells are {\em blocked}, a starting cell $c_0 = (x_0, y_0)$ and a sequence of $m$ {\em jumps} $\Delta_1,\Delta_2,\dots,\Delta_m$, $\Delta_i = (dx_i, dy_i), -n \leq dx_i, dy_i \leq n$ ($m$ equals the number of the initial empty cells).\\

\noindent{\bf Question}: Does there exist a sequence of integers $(s_1,s_2,\dots,s_m), s_i \in \{-1,+1\}$ such that if a frog is placed on the starting cell, it can {\em visit} (we will equivalently use the term {\em fill}) every empty cell of the board exactly once following the sequence of jumps:
$$(x_i, y_i) = (x_{i-1} + s_i * dx_i, y_{i-1} + s_i * dy_i)$$
(i.e. the frog can choose only the \emph{direction} of the given jumps)?
The frog cannot jump outside the board, on a blocked cell or on an already
visited cell.
\end{definition}

\subsection{Terminology}

We briefly introduce the terminology used in the next sections.

\begin{description}
\item[Horizontal jump]: a jump of the form $(\Delta x,0)$ ({\em horizontal step} if $\Delta x=1$);
\item[Vertical jump] : a jump of the form $(0,\Delta y)$ ({\em vertical step} if $\Delta y=1$);
\item[Diagonal jump]: a jump of the form $(\Delta x,\Delta x)$ ({\em diagonal step} if $\Delta x=1$);
\item[Sequence of jumps]: a (sub)sequence $\Delta_i,\Delta_{i+1},\dots,\Delta_m$ of jumps that are part of the input; the frog at each step $\Delta_i$ must choose a direction ($s_i \in \{\pm 1\}$) and make the jump ($+\Delta_i$ or  $-\Delta_i$).
\end{description}

When there is no ambiguity  horizontal/vertical/diagonal jumps are abbreviated with a single letter
or a single number; for example if we are describing a horizontal sequence of jumps, we can write $2,3,2$ instead of $(2,0),(3,0),(2,0)$, or write $x_1,x_2,\dots,x_j$ instead of $(x_1,0),(x_2,0),\dots,(x_j,0)$.

\begin{description}
\item [Line]: a row of the board; we use a string to represent it using
these characters: $B$=blocked cell, $E$=empty cell, $F$=frog, $V$=visited cell; repeated cells can be represented with power notation (e.g. $B^2FE^3 = BBFEEE$);
\item [Strip]: two or more lines;
\item[Configuration]: the status of the board cells ($B$,$E$ or $F$) and the sequence of the next jumps; a configuration is {\em valid}, if the frog can complete the sequence of jumps choosing a $\pm 1$ direction for each of them;
\item[Gadget]: the configuration of a part of the board that have a particular role in the  reduction; a gadget can have one or more {\em entrance cells}
and one or more {\em exit cells}.
A {\em (valid) traversal} of the gadget is a sequence 
of jump directions that can lead the frog from an entrance cell to an exit cell.
We will use capital letters: (e.g. $L,C,S$) to indicate gadgets, capital letters
with a tilde (e.g $\tilde{L},\tilde{C},\tilde{S})$ to indicate jump sequences.

\end{description}

In the figures we will use the following notation: gray cells represent blocked cells; a cell with a $F$ represents the frog position; a cell with a $t$ represents
the target cell; a cell with a $V$ represents 
an already visited cell; numbered cells represent a (valid) sequence of jumps that the frog can make. The horizontal coordinates $x = 0,1,2,\ldots$
are from left to right, the vertical coordinates $y = 0,1,2,\ldots$ are
from top to bottom.

\section{$\sf{NP}$--Completeness}
\label{sec:npc}

We will prove that the CFP problem is $\sf{NP}$--hard giving a polynomial time
reduction from
the Hamiltonian path problem on grid graphs \cite{ips82}.
First we underline a general property that we will use in the gadget construction.

\begin{lemma}[Gadget construction]
Given a sequence of jumps $(x_1,y_1),\dotsc,(x_m,y_m)$ and
a $w \times h$ rectangular area $R$ of the board, we can construct a corresponding gadget
in which all
jumps of the sequence must be made inside it and the frog can exit it only at the end.
\end{lemma}

\begin{proof}
It is sufficient to extend the area adding a $B^{\max\{x_i\}}$ border of
blocked cells on the left/right, a $B^{\max \{y_i\}}$ 
border of blocked cells on the top/bottom.
Then we can make the frog enter the gadget with a vertical jump $y_{In} > \max \{y_i\}$ and leave it
with another vertical jump $y_{Out}> \max \{y_i\}$. The extended sequence of jumps for the traversal is:

$$\dots,y_{In}, (x_1,y_1),\dots,(x_n,y_n), y_{Out},\dots$$
\end{proof}

Note that \emph{a)} if $R$ contains $m+1$ empty cells (the number of jumps
is equal to the number of empty cells minus one), then
a valid traversal of the gadget implies that the frog must visit (fill) all its
empty cells; \emph{b)} instead of vertical jumps we can use long enough horizontal or diagonal jumps. Figure~\ref{fig:gadgetcons} shows an example
of a $5\times 5$ partially filled region and a sequence of 6 jumps;
the region can be embedded in a $7 \times 7$ gadget that can be traversed in 4
different ways; all traversals completely fill the original inner region.   

\begin{figure}[htp]
\centering
\includegraphics[width=12cm]{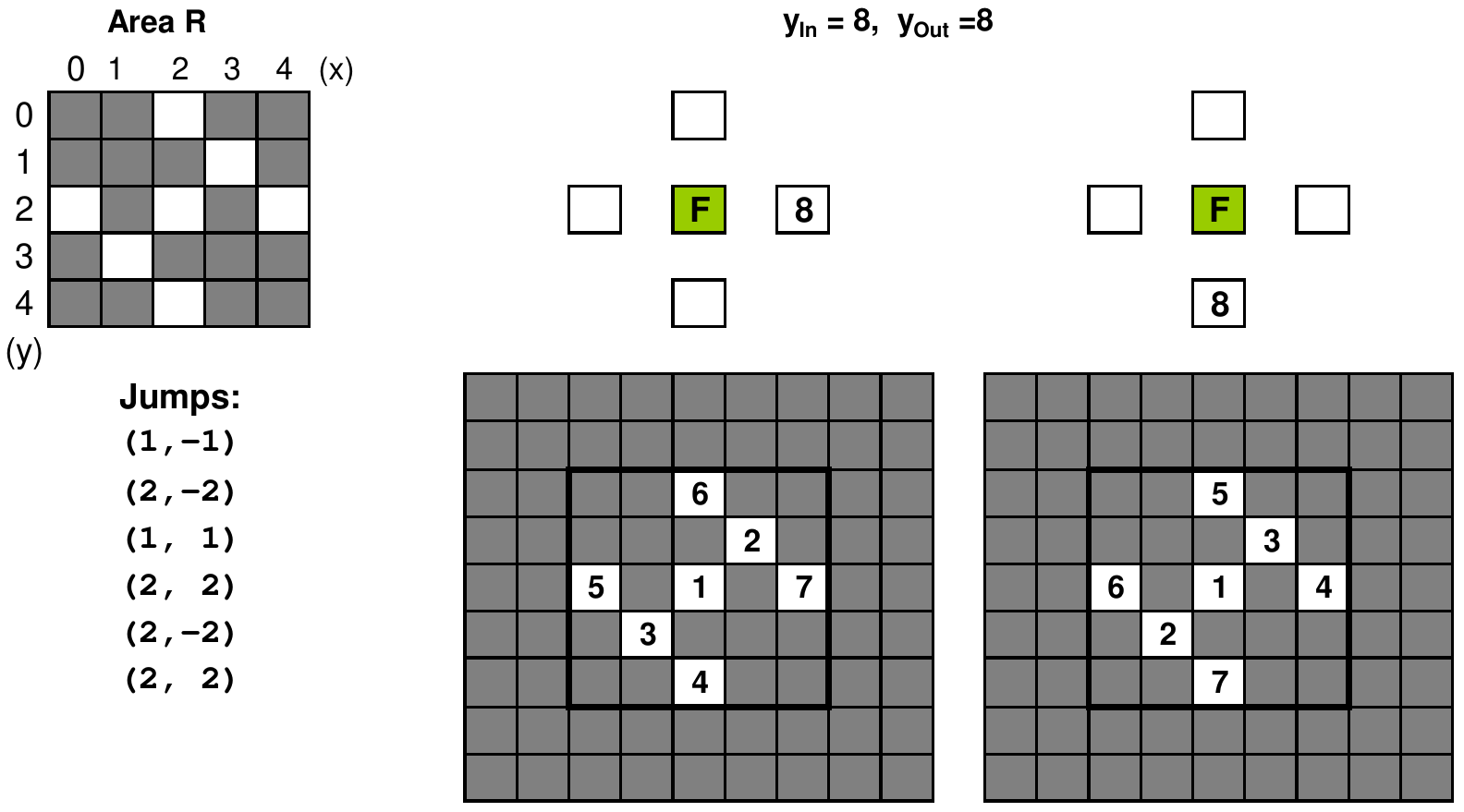}
\caption{A $5\times 5$ partially filled region and a sequence of 6 jumps;
the region can be embedded in a $7 \times 7$ gadget; two of
the four possible traversals are showed on the right.}\label{fig:gadgetcons}
\end{figure}

\subsection{Reduction overview}

The reduction is from the $\sf{NP}$--complete Hamiltonian path problem on grid graphs 
(which may also contain holes) \cite{ips82}. Let $G^\infty$ be the infinite graph whose vertex set consists
of all points of the plane with integer coordinates
and in which two vertices are connected if and only if the (Euclidean) distance between
them is equal to 1. A \emph{grid graph} is a finite, node--induced subgraph of $G^\infty$.

Given an $m \times m$ grid graph $G$ in which $|V|=n$,  $s,t \in V$ are the source and target nodes; and the coordinates of node $u_i, 1\leq i  \leq n$ on the grid are $(x_{u_i},y_{u_i})$.

Pick the first $k \geq 2$ such that $2^k \geq 4m$ and build a
$2^k-1 \times 2^k-1$ board $R$  with all the cells blocked except
the cells at coordinates $(4x_{u_i},4y_{u_i})$   corresponding to the {\em nodes} of
$G$ and the empty
{\em target cell} $(4x_{t}+1,4y_{t})$ one step aside from the node $t$.
We call this part of the board the {\em graph area}. The frog is initially positioned on cell $(4x_s, 4y_s)$.

For brevity, throughout the paper we will denote its size with $w=2^k-1$ and
$v = 2^{k-1} = \lceil w/2 \rceil$.

We extend the board at the bottom with $n-1$ {\em edge gadgets} $L_1,L_2,\dots,L_{n-1}$; each edge gadget $L_i$
has an associated {\em cleanup gadget} $C_i$ placed on its right. We will generate a fixed sequence of jumps that will force the following logical \emph{ phases}:

\begin{itemize}
\item ${[\tilde{L}_1]}$ enter gadget $L_1$, choose one of the four directions up, down, left, right and return to cell
$(x_0 \pm 4, y_0)$ or $(x_0 , y_0\pm 4)$ in the graph area;
\item ${[\tilde{L}_2]}$ enter gadget $L_2$, choose one of the four directions and return to cell
$(x_1 \pm 4, y_1)$ or $(x_1 , y_1\pm 4)$ in the graph area;
\item \dots
\item ${[\tilde{L}_{n-1}]}$ enter gadget $L_{n-1}$, choose one of the four directions and return to cell
$(x_{n-1} \pm 4, y_{n-1})$ or $(x_{n-1} , y_{n-1}\pm 4)$ in the graph area;
\item ${[\tilde{T}]}$ jump to the target cell $(x_t,y_t)$;
\item enter the cleanup gadgets area;
\item ${[\tilde{C}_1]}$ completely fill the lines of the edge gadget $L_1$ that have an already  visited cell (visited during phase $\tilde{L}_1$); then completely fill the lines  that are still empty;
\item \dots
\item ${[\tilde{C}_{n-1}]}$ completely fill the lines of the edge gadget $L_{n-1}$ that have an already  visited cell (visited during phase $\tilde{L}_{n-1}$); then completely fill the lines  that are still empty.
\end{itemize}

An outline of the whole board is shown in Figure~\ref{fig:outline}.


\subsection{Edge gadgets}

Each edge gadget $L_i$ is a rectangular area that has the same width $w$ of the graph area, and $7w$ lines:
$2w$ blocked lines; a first {\em top inner strip}
of height $w$ in which even lines are empty and odd lines are blocked;
$w$ blocked lines;
a second {\em bottom inner strip} of height $w$
 in which even lines are empty and odd lines are blocked;
and finally $2w$ blocked lines. If the gadget is positioned at
row $l_i$ its structure is:

\begin{figure}[H]
\centering
\includegraphics[width=14cm]{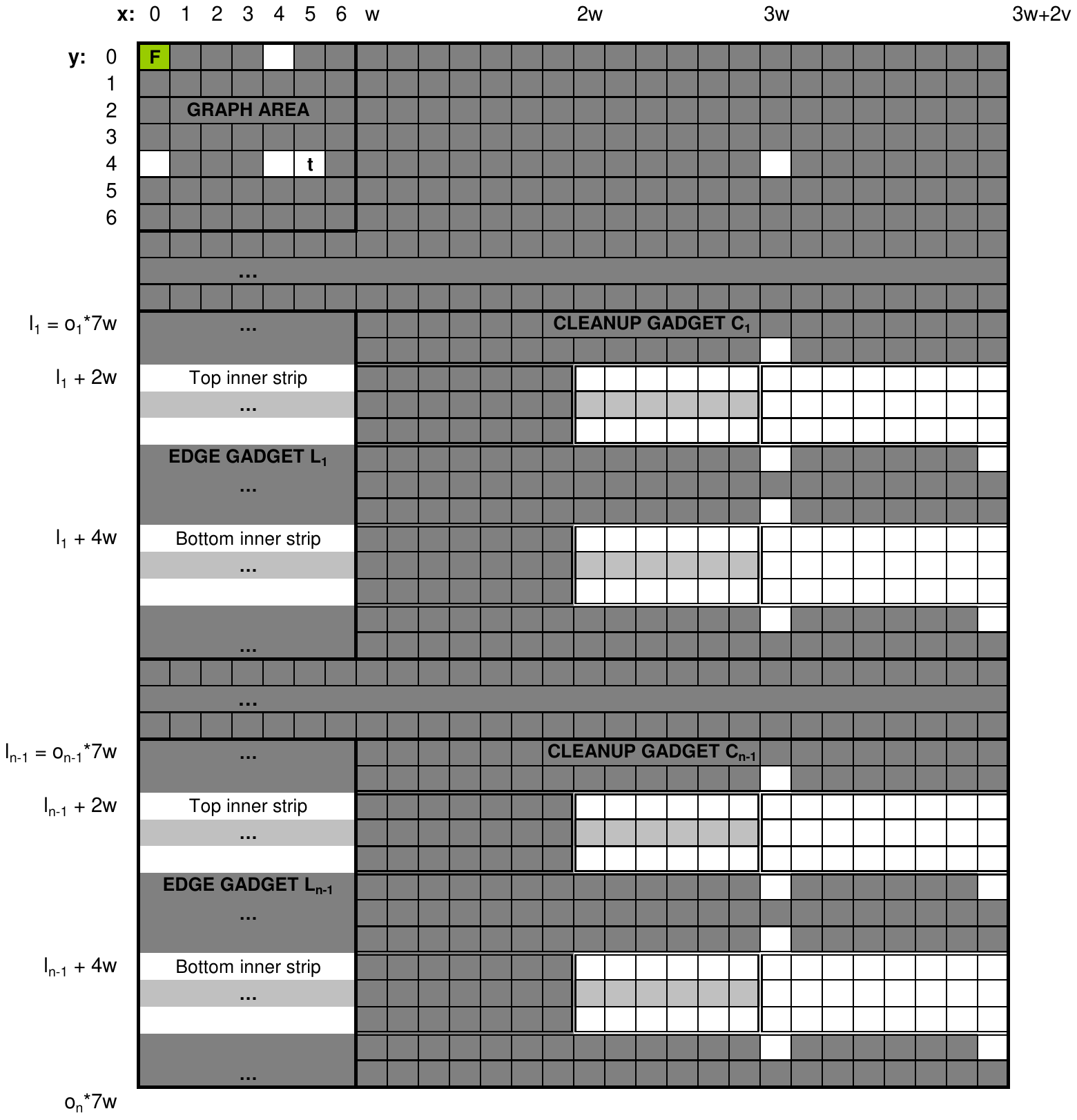}
\caption{An outline of the board generated by the reduction.}\label{fig:outline}
\end{figure}

\begin{center}
\begin{tabular}{ r c l  }
row \#& cells & repetitions\\\hline
$l_i$: & $B^w$ & $\times 2w$ \\\hline
\specialcell{$l_i+2w$:\\~} & \specialcell{$E^w$\\$B^w$} & $\times \lfloor \frac{w}{2} \rfloor$\\
 & $E^w$ & $\times 1$\\\hline
$l_i + 3w$: & $B^w$ & $\times w$\\\hline
\specialcell{$l_i+4w$:\\~} & \specialcell{$E^w$\\$B^w$} & $\times \lfloor \frac{w}{2} \rfloor$\\
 & $E^w$ & $\times 1$\\\hline
$l_i+5w$: & $B^w$ & $\times 2w$ \\\hline\\
\end{tabular}
\end{center}

We make the frog enter an edge gadget from the graph area
 with a vertical jump $J_{L_i}$  on an empty cell of the upper inner strip,
and leave it from the bottom inner strip
with a vertical jump $J'_{L_i} = J_{L_i}+2w$ that make it return to
the graph area. 

The vertical positions $l_i$ of the edge gadgets must be chosen in such a way 
that the frog cannot leave one of them and directly jump to another edge gadget,
but is forced to return to the graph area.
This is achieved using the vertical positions 
$l_i = o_i * 7w$,  for odd $o_i \geq 1$,
and setting $J_{L_i} = l_i + 2w$ and $J'_{L_i} =  J_{L_i}+2w = l_i + 4w$.
Indeed if the frog begins on $(4x,4y)$ in the graph area for some $k$,
it'll end up at vertical coordinate $l_i + 4w + 4y + 4z$ for $z$ in $\{-1,0,1\}$
at the end of the traversal of gadget $L_i$. At this point the next vertical jump
of $(0,l_i + 4w)$ shouldn't allow the frog to jump further down to another
edge gadget; i.e. $l_i + 4w + 4y + 4z + l_i + 4w = 2l_i+8w+4y+4z$
should be different from $l_j + 2w + 2a$ and $l_j + 4w + 2b$ for all $i \neq j$ and
$0 \leq a,b < w/2$; but the inequalities $2l_i+8w+4y+4z \neq l_j + 2w + 2a$ and
$2l_i+8w+4y+4z \neq l_j + 4w + 2b$ hold because the left-hand side is even,
and the right-hand side is an odd number ($l_j = o_j * 7w$ is odd because both $o_j$ and
$w$ are odd) plus an even number.


The sequence of jumps inside each edge gadget is:

$$\tilde{L}_{seq} = (2,2), (0,2w), (2,-2)$$

The first jump is a diagonal jump that must be made in the top inner strip,
the second jump forces the frog to jump to the bottom inner strip,
the third jump is a diagonal jump that must be made in the bottom inner strip.
The traversal of gadget $L_i$ is forced with the sequence of jumps:

$$\tilde{L}_{i} = (0,J_{L_i}) , \tilde{L}_{seq} , (0,J'_{L_i})$$

Suppose that the frog is on cell $(x,y)$ and enters the gadget $L_i$ from the top with the vertical jump $(0,J_{L_i})$:
after the $\tilde{L}_{seq}$ jumps it can only use the final vertical
jump $(0,J'_{L_i})$
to return to the graph area, and its final position must be one of the cells: $(x+4,y), (x-4,y), (x,y-4), (x,y+4)$.
After each traversal only four cells of the inner strips are visited.  
Figure~\ref{fig:edgegad3} shows an example of a $15 \times 60$ edge gadget ($w=15$) and its possible traversals.


We can define the first part of the jump sequence of the CFP in this way:

\begin{center}
\begin{tabular}{c | c | p{5cm} }
Phase & Jumps & \\\hline
$\tilde{L}_1$ & $(0,J_{L_1}), \tilde{L}_{seq}, (0,J'_{L_1})$ & enter edge gadget $L_1$, traverse it and return to graph area;\\
$\tilde{L}_2$ & $(0,J_{L_2}), \tilde{L}_{seq}, (0,J'_{L_2})$ & enter edge gadget $L_2$, traverse it and return to graph area;\\
\dots  & \dots & \dots\\
$\tilde{L}_{n-1}$ & $(0,J_{L_{n-1}}), \tilde{L}_{seq}, (0,J'_{L_{n-1}})$ & enter edge gadget $L_{n-1}$, traverse it and return to graph area;\\
$\tilde{T}$ & (1,0) & jump to target cell $t$.\\\hline
\end{tabular}
\end{center}

\begin{figure}[H]
\centering
\includegraphics[width=\textwidth]{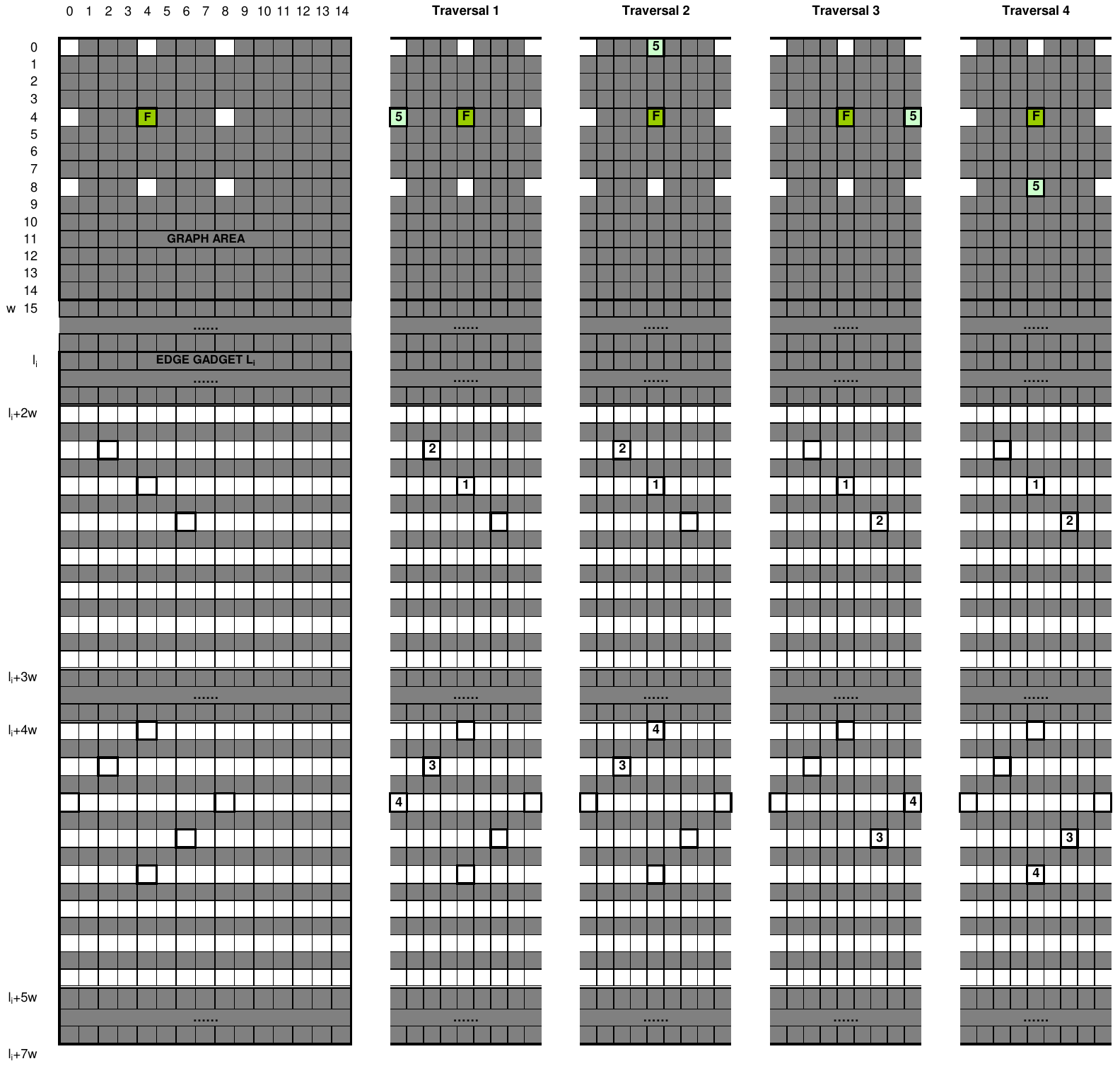}
\caption{The edge gadget and its four possible traversals.}\label{fig:edgegad3}
\end{figure}

Note that the final odd horizontal step (the only odd horizontal step), forces the frog to choose a path in the graph area in which
the final cell is the one corresponding to node $t$. Furthermore if the frog
is on cell $(x_i,y_i)$ and traverses the edge gadget $L_i$ it must, by construction,
be in one of the four adjacent cells $(x_i\pm 4,y_i),(x_i,y_i \pm 4)$ and that
cell must be empty (i.e. correspond to an unvisited node),
so the sequence of cells visited in the graph area must 
correspond to an Hamiltonian path from $s$ to $t$ on the original graph $G$.
So the following lemma holds:

\begin{lemma}The frog can reach cell $(x_t,y_t)$ from its initial position
$(x_s,y_s)$ if and only if there is an Hamiltonian path from $s$ to $t$ in the original graph $G$.
\end{lemma}

At the end of the graph area traversal, most cells of the edge gadgets are still empty, so we must extend the jump sequence to let the frog visit all of them and completely fill the board.

The cleanup gadgets are more complicated because they must allow the frog
to fill both the lines of the edge gadgets that has a single blocked cell, and the lines
of the edge gadgets that have been left empty.

\subsection{Cleanup gadgets}

Every cleanup gadget $C_i$ 
has two similar \emph {strip cleanup gadgets} $S_i^{a}$,$S_i^{b}$  one for the top 
and one for the bottom inner strip
of the corresponding edge gadget $L_i$.
The strip cleanup gadget is placed
on the right of the corresponding inner strip, at coordinate
$(w,l_i+2w)$ for top inner strips ($(w,l_i+4w)$ for bottom inner strips), and is
 a $3w+2v \times w$ rectangular area of the board
(note that $v = \lceil w/2 \rceil$ is simply the number of even rows
in the inner strip)
with the following structure:

\begin{center}
\begin{tabular}{ r c l  }
row \#& cells & ~ \\\hline
even rows: & $B^w E^w E^{2v}$ & ~ \\
odd rows: & $B^w B^w E^{2v}$ & ~ \\

\end{tabular}
\end{center}
Furthermore, outside the gadget at coordinates $(3w,l_i+2w-1),(3w,l_i+3w),(3w+2v-1,l_i+3w)$ there are three empty cells that are the entrance
and exit cells of the gadget $S_i$. Figure~\ref{fig:cleanup2} shows the
outline of
a strip cleanup gadget associated to a $7 \times 7$ ($w=7$) inner strip. 

\begin{figure}[htp]
\centering
\includegraphics[width=\textwidth]{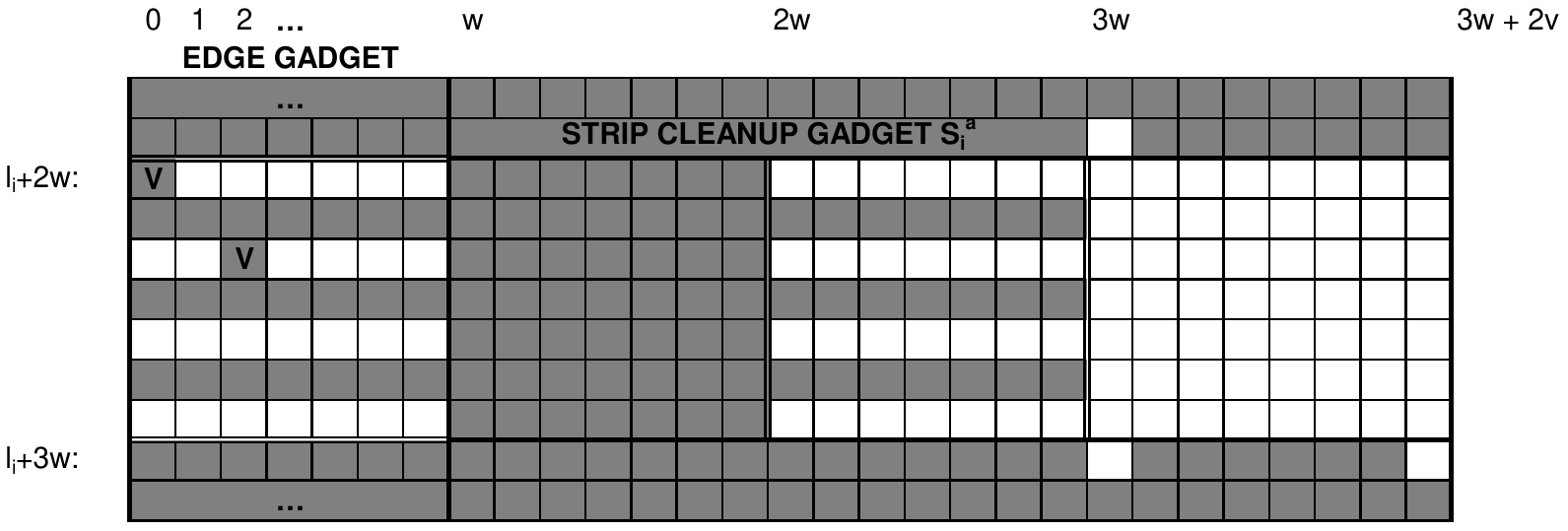}
\caption{A strip cleanup gadget $S_i$ associated to the top inner strip of the edge gadget $L_i$ (assuming $w=7$).}\label{fig:cleanup2}
\end{figure}

The jump sequence $\tilde{C_i}$ using for traversing the cleanup gadget
$C_i$ is:
$$\tilde{C_i} = \tilde{S_i},(-2v+1,),(0,w-1),\tilde{S_i}$$
where  $\tilde{S_i}$ is the jump sequence for traversing a strip
cleanup gadget and $(-2v+1,),(0,w-1)$ are the jumps from the exit of the top
strip cleanup gadget to the entrance of the bottom strip cleanup gadget. 

The jump sequence $\tilde{S_i}$ that allows the frog to traverse the strip
cleanup gadget $S_i$ has the following components:

\begin{itemize}

\item $v$ \emph{vertical selector sequences};
that allow the frog to choose an even row of the inner strip;

\item $2$ \emph{horizontal hole sequences}
that allow the frog to completely fill the two even rows of the inner strip that have an already visited cell;

\item $v - 2$ \emph{horizontal fill sequences}
 that allow the frog to completely fill the
remaining $v-2$ even empty rows of the inner strip.

\end{itemize}

The horizontal hole and fill sequences are embedded in the
vertical selector sequences. 

\subsection{Horizontal sequences}

We see how to build the two horizontal hole sequences that can
be used to fill the lines of the edge gadgets that are unvisited
or that contain an already visited cell.

\begin{lemma}[Horizontal fill sequence]
If $w = 2^k - 1$,
starting from the line:
$$E^wB^wE^wFE$$ we can force the frog to correctly visit and fill all the empty cells and finally jump on the rightmost cell;
i.e. the final configuration is: $$V^wB^wV^wVF$$
\end{lemma}

\begin{proof}
We can use the the following horizontal jump sequence:
$$\tilde{Z} = 3w, \overbrace{1,1,\dots,1}^{w-1\mbox{ times} }, w+1,  \overbrace{1,1,\dots,1}^{w-1\mbox{ times}}  ,2$$
It forces the frog to jump to the leftmost empty cell,
fill the $w$ empty cells of the first $E^w$ block, jump to the second $E^w$ block,
fill it and finally jump on the rightmost cell.
\end{proof}

\begin{figure}[htp]
\centering
\includegraphics[width=\textwidth]{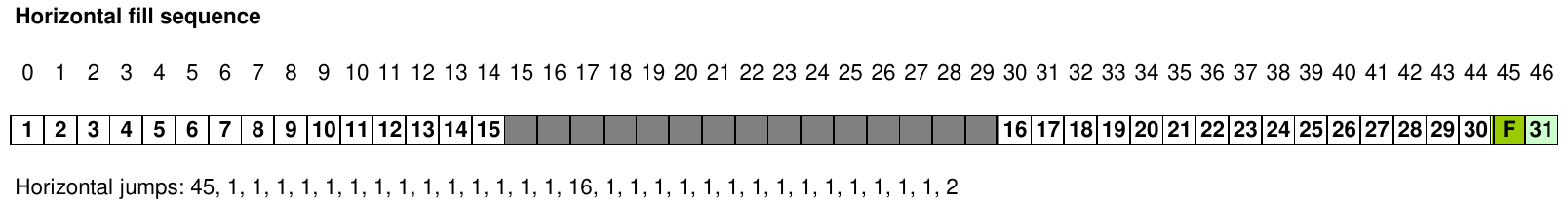}
\caption{An horizontal fill sequence on an inner strip of width
$w=15$ ($k=4$).}\label{fig:horizfill}
\end{figure}

Figure~\ref{fig:horizfill} shows an horizontal fill sequence on an inner strip of width
$w=15$ ($k=4, v=8$). We can also build an horizontal sequence
that completely fills an inner strip row that contains an already visited
cell (i.e. leaves a ``hole''). First we prove the following:

\begin{lemma}[Binary jump sequence]\label{binarylemma}
If $w = 2^k - 1, v = 2^{k-1},\, k \geq 2$, starting from the line: 
$$E^{v-1} F E^{v-1}$$ 
i.e. the frog is placed in the middle of an empty block of length $w$, we can force the frog to correctly fill the line and end
in an arbitrary even position;
(i.e. the final configuration is: $$V^{a-1}FV^{w-a}$$ (where $a$ is even) using 
a fixed binary jump sequence:
$$\tilde{U}_{k-1}, \tilde{U}_{k-2}, \dots,  \tilde{U}_{2}, \tilde{U}_{1}$$
where $\tilde{U}_j$ is the sequence of horizontal jumps: $\overbrace{1,1,\dots,1}^{2^j-1 \mbox{ times}},2^{j}+2^{j-1}-1$.
\end{lemma}

\begin{proof}
We can prove it by induction on $k$: when $k = 2$ the line is $EFE$, and
the binary jump sequence is $\tilde{U}_{1} = 1,2$; so the frog can make the first
jump left or right and then reach one of
the two possible final even cells ($VVF$ or $FVV$).
Now, suppose that the statement holds for $k$; and the frog is placed in
the middle of a $w' = 2^{k+1}-1$ empty block ($E^{2^k} F E^{2^k}$).
The first jump sequence is\\ $\tilde{U}_{k} = \overbrace{1,1,\dots,1}^{2^k-1 \mbox{ times}},2^{k}+2^{k-1}-1$; suppose that the target even cell
is on the right half of the line, the frog can make $2^k-1$
left horizontal jumps and completely fill the left half of the line
($FV^{2^k}E^{2^{k}-1}$); then the final (right) jump $2^{k}+2^{k-1}-1$ can 
lead it in the middle of the right half of the line that contains the target even cell ($V^{2^k} E ^{2^{k}-1} F E ^{2^{k}-1}$); and by induction hypothesis
it can fill the right half and end in the target even cell.
If the target even cell was on the left half of the line,
the frog can simply invert the jump directions.
Informally the line can be seen as a binary tree with
the leaves corresponding to even cells; the sequence of jumps
allows the frog to choose a half of the tree, fill it, jump to the other half,
fill it and so on until the tree is completely visited.
\end{proof}

\begin{figure}[htp]
\centering
\includegraphics[width=\textwidth]{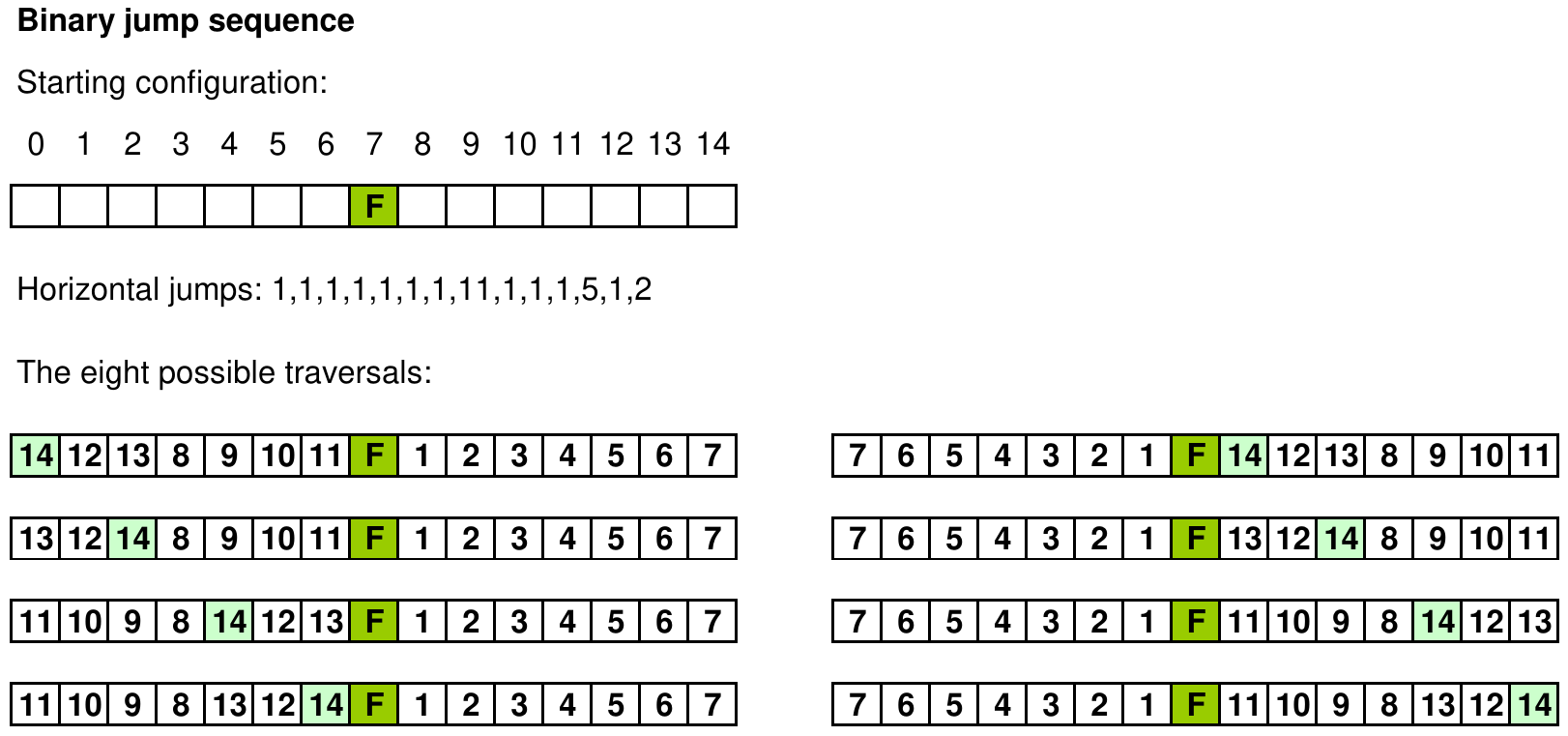}
\caption{A binary jump sequence on a line of width
$w=15$ ($k=4$) and the jumps that place the frog on cells 0,2,4,6,8,10,12,14.}\label{fig:binseq}
\end{figure}

Figure~\ref{fig:binseq} shows some examples of a binary jump sequence on a line of width
$w=15$ ($k=4$).
If we reverse all the sequences of jumps we can go from the final configuration of
Lemma~\ref{binarylemma} to the initial configuration, so the following also holds:

\begin{lemma}[Reverse binary jump sequence]\label{binarylemmarev}
If $w = 2^k - 1, v = 2^{k-1},\, k \geq 2$ and the frog is placed on an arbitrary even cell of an empty
line of length $w$ ($E^{a-1}FE^{w-a}$), then the following fixed sequence of jumps
that allows the frog to completely fill the line and end in the center cell
($V^{v-1} F V^{v-1}$):
$$\tilde{U}_{1}^R, \tilde{U}_{2}^R, \dots, \tilde{U}_{k-2}^R, \tilde{U}_{k-1}^R$$
where $\tilde{U}_j^R$ is the reverse of the sequence $\tilde{U}_j$ of horizontal jumps:\\
$\tilde{U}_j^R = 2^{j}+2^{j-1}-1,\overbrace{1,1,\dots,1}^{2^j-1 \mbox{ times}}$
\end{lemma}


Combining the sequences used in the previous two lemmas we can build the horizontal
hole sequence:

\begin{lemma}[Horizontal hole sequence]\label{lem:hhole}
If $w = 2^k - 1, v = 2^{k-1}, k\geq 2$, starting from the line: 
$$E^w B^w E^w F E$$ we can force the frog to correctly visit and fill the line except one empty (or visited) cell in an arbitrary even position of the leftmost empty block and finally jump on the rightmost cell;
i.e. the final configuration is: $$(V^{a-1} E V^{w-a}) B^w V^w V F$$ where $a$ is even.
\end{lemma}

\begin{proof}
We use the following sequence of horizontal jumps:

$$\tilde{H} = 2w+v, \tilde{U}_{k-1}, \tilde{U}_{k-2}, \dots , \tilde{U}_2, 1,2w, \tilde{U}_{1}^R, \tilde{U}_2^R, \tilde{U}_3^R, \dots, \tilde{U}_{k-1}^R, v+1$$

Where $\tilde{U}_j$ is the sequence of horizontal jumps: $\overbrace{1,1,\dots,1}^{2^j-1 \mbox{ times}},2^{j}+2^{j-1}-1$; and
$\tilde{U}_j^R$ is the reverse of $\tilde{U}_j$.

The first jump ($2w+v$) forces the frog in the middle of the first $E^w$ block, then
it can follow the fixed sequence of jumps of Lemma~\ref{binarylemma} except for the
very last jump; note that the
sequence of Lemma~\ref{binarylemma} ends with $\dots,\tilde{U}_2,\tilde{U}_1$
while the sequence here is $\dots,\tilde{U}_2,1,\dots$, i.e. it misses the final
horizontal jump of $\tilde{U}_1$ that has length $2$ ($\Delta x = 2$).
In this way we can reach a configuration in which an arbitrary even cell $(2x,\cdot)$ is
left empty, and the frog is positioned at distance $2$ from that cell, i.e. it is
on the even cell $(2x+2,\cdot)$ or $(2x-2,\cdot)$.
So it can jump on the second $E^w$ block 
with the horizontal jump $2w$ and land on the even cell $(2x+2+2w,\cdot)$ or
$(2x-2+2w,\cdot)$;
and then it can completely fill the block reaching its center cell by using
the fixed sequence of jumps of Lemma~\ref{binarylemmarev}.
Finally it can make a final $v+1$ jump to reach the rightmost empty cell.
\end{proof}

\begin{figure}[htp]
\centering
\includegraphics[width=12.5cm]{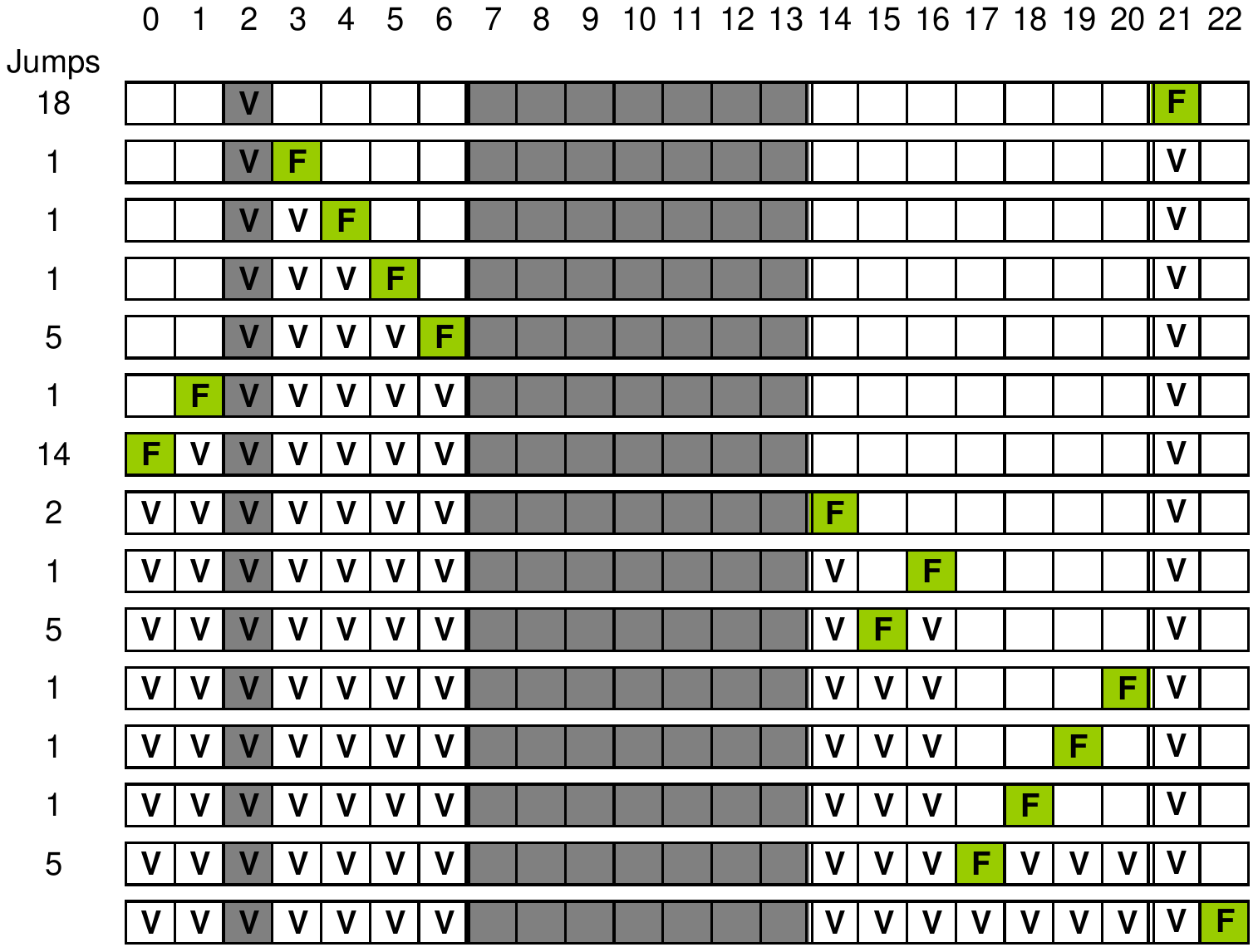}
\caption{A possible horizontal hole sequence traversal that
fills an inner strip row with an already visited cell.}\label{fig:hhole}
\end{figure}

Figure~\ref{fig:hhole} show a possible traversal of the line $E^{7}B^{7}E^{7}FE, (k=3, w=7, v=4)$ that leaves a hole (corresponding
to an already visited cell) in even position
using the sequence of jumps defined in Lemma~\ref{lem:hhole}.

\begin{figure}[htp]
\centering
\includegraphics[width=\textwidth]{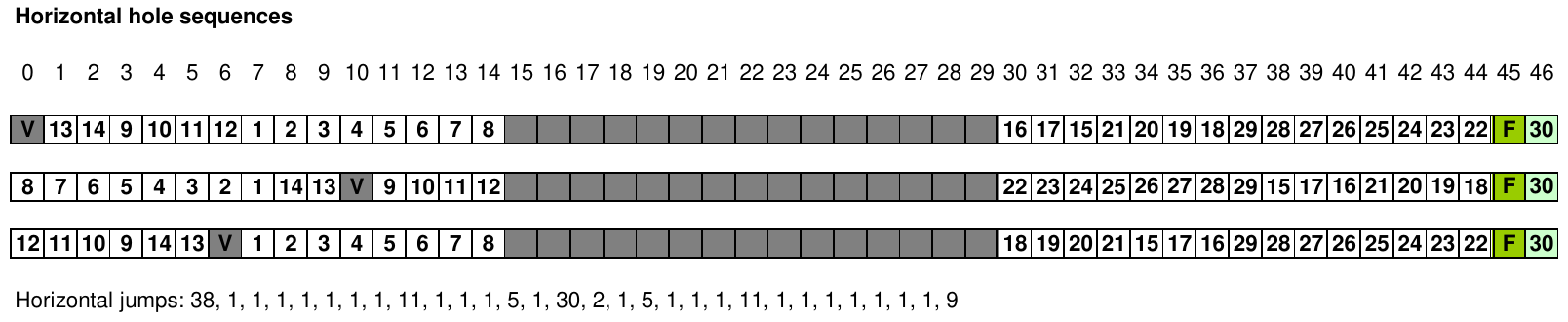}
\caption{Three horizontal hole sequences
with the already visited cell in different positions 
and a horizontal fill sequence on an inner strip of width
$w=15$ ($k=4, v=8$).}\label{fig:hhole2}
\end{figure}

Figure~\ref{fig:hhole2} shows three horizontal hole sequences
with the already visited cell in different positions on an inner strip of width
$w=15$ ($k=4, v=8$).

Note that in both types of horizontal sequences (fill and hole), we can
extend the first and last steps so that the frog can be
positioned farther from the inner strips.


\subsection{Vertical selector sequence}

Using a similar sequence of the horizontal hole sequence,
we can build a sequence of jumps that allows the frog, placed on
the top of a $2v \times w$ empty area to 
vertically select a different even row, $v$ times and finally exit the area on 
the cell at the bottom-left. 

\begin{lemma}[Vertical selector sequence]
If $w = 2^k - 1, v = 2^{k-1}, k \geq 2$, starting from the $2v \times w+2$ configuration:

\begin{center}
\begin{tabular}{r c }
row \# & cells\\\hline
$l_i +2w -1:$ & $FB^{2v-2}B$  \\\hline
\specialcell{$l_i+2w$:\\~\\~} & \specialcell{$E^{2v}$\\\dots\\$E^{2v}$}\\\hline
$l_i +3w:$ & $EB^{2v-2}E$  
\end{tabular}
\end{center}

we can allow the frog to choose a different even row $v$ times, and reach the configuration:

\begin{center}
\begin{tabular}{r c }
row \# & cells\\\hline
$l_i +2w -1:$ & $VB^{2v-2}B$  \\\hline
\specialcell{$l_i+2w$:\\~\\~} & \specialcell{$V^{2v}$\\\dots\\$V^{2v}$}\\\hline
$l_i +3w:$ & $FB^{2v-2}V$  
\end{tabular}
\end{center}

\end{lemma}

\begin{proof}
The sequence
is:
$$\tilde{V} = (0,v), \tilde{V'}_1, \overline{(1,0)}, \tilde{V}''_1, (1,0), \dots, (1,0), \tilde{V'}_v, \overline{(1,0)}, \tilde{V}''_v,(0,v),(2v-1,0)$$
Where: 
$$\tilde{V'}_i =  \tilde{U}_{k-1}, \tilde{U}_{k-2}, \dots , \tilde{U}_2, \tilde{U}_1$$
$$\tilde{V''}_i =  \tilde{U}_1^R,\tilde{U}_2^R,\tilde{U}_3^R,\dots,\tilde{U}_{k-1}^R$$
$\tilde{U}_j$ is the sequence of vertical jumps: $\overbrace{1,1,\dots,1}^{2^j-1 \mbox{ times}},2^{j}+2^{j-1}-1$; and
$\tilde{U}_j^R$ is the reverse of $\tilde{U}_j$.
In other words we are using the sequence of jumps of Lemma~\ref{binarylemma}
and Lemma~\ref{binarylemmarev}, but here the jumps are vertical.
The (vertical) sequence $\tilde{U}_j$ selects a cell on an even row, then, after
an horizontal right jump $(1,0)$, the reverse $\tilde{U}_j^R$ allows the frog to
reach the center row again.
\end{proof}

Figure~\ref{fig:vertsel} shows the possible traversal of a $7 \times 8$
area using the vertical selector sequence.

\begin{figure}[H]
\centering
\includegraphics[width=\textwidth]{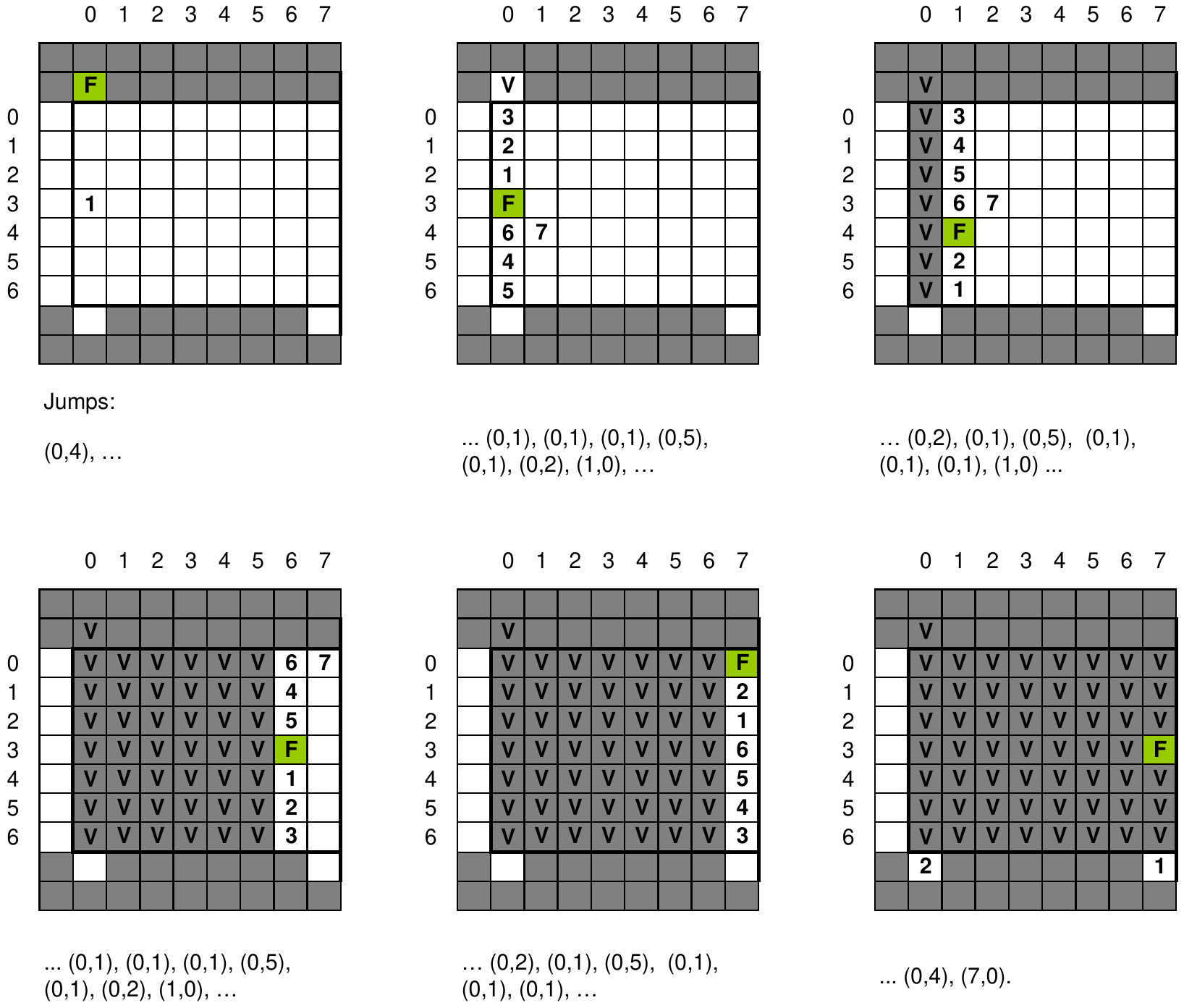}
\caption{An example of the vertical selector sequence
on a $7 \times 8$ area ($w=7, v=4$).}\label{fig:vertsel}
\end{figure}

\subsection{Linking the gadgets}

The horizontal inner jumps $\overline{(1,0)}$ in the vertical selector sequences $\tilde{V}$
can be replaced by an horizontal hole sequence or by an horizontal fill 
sequence (extending their first and last jump).
So we can build a sequence of jumps that allows the frog to:
\emph{a)} select 2 even rows of the inner strip with a visited cell and fill them;
\emph{b)} select the remaining $v-2$ empty even rows and fill them.

The complete inner strip sequence $S_i$ is:

$$\begin{array}{c l}
\tilde{S}_i = & (0,v), \tilde{V'}_1, \tilde{H}, \tilde{V}''_1, (1,0), \tilde{V'}_2, \tilde{H}, \tilde{V}''_2, (1,0),\\
~ & \tilde{V'}_3, \tilde{Z}, \tilde{V}''_3, (1,0),\dots,
 \tilde{V'}_{v}, \tilde{Z}, \tilde{V}''_v, (0,v),(2v-1,0)
\end{array}$$

Figure~\ref{fig:final} shows a possible traversal of the strip
cleanup gadget of Figure~\ref{fig:cleanup2}.
\begin{figure}[htp]
\centering
\includegraphics[width=\textwidth]{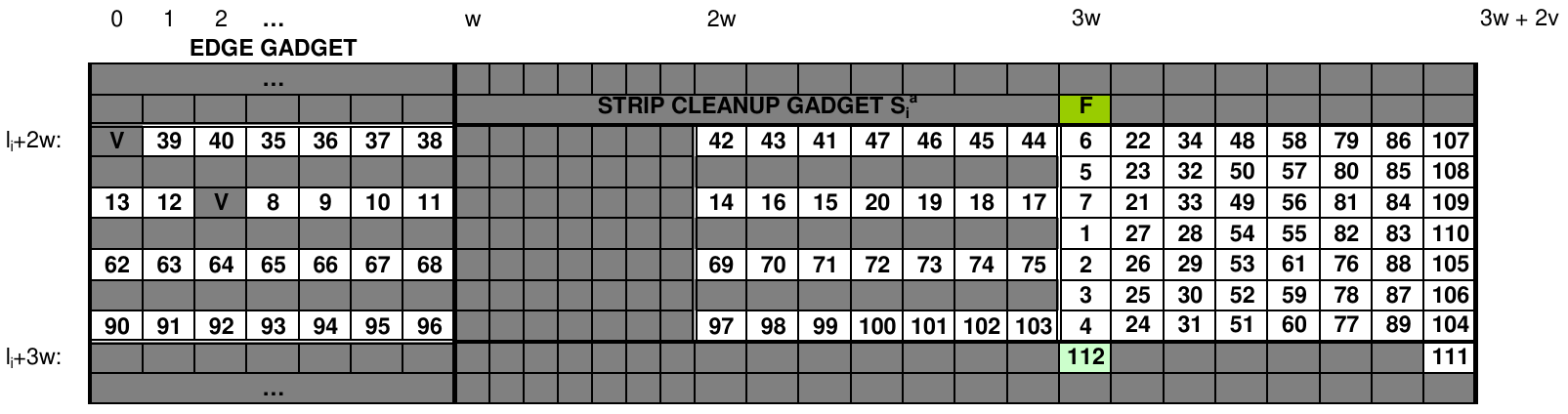}
\caption{A possible traversal of the strip
cleanup gadget of Figure~\ref{fig:cleanup2},
that combines the vertical selector sequence with the horizontal sequences.}\label{fig:final}
\end{figure}

Finally we can link together the cleanup sequences $\tilde{C}_i$;
first we add an empty cell at coordinates $(3w,y_t)$ that
allows the frog to jump from the target cell $(x_t,y_t)$ to 
the entrance cell
of the top inner strip gadget $S_i^a$; then we add
a vertical jump from gadget $C_i$ to gadget $C_{i+1}$, $i = 1,\dots,n-2$:

$$\begin{array}{c l}
\tilde{C} = & (3w-x_t,0),(0,l_1+2w-1-y_t),\tilde{C}_1,(0,l_2+2w-1-l_1+5w), \dots\\
~ & \dots,(0,l_{n-1}+2w-1-l_{n-2}+5w), \tilde{C}_{n-1}
\end{array}$$

\begin{theorem}
The {\sc Crazy Frog Puzzle} is $\sf{NP}$--complete.
\end{theorem}

\begin{proof}
The problem is $\sf{NP}$--hard: 
by construction if the original grid graph problem has a solution, then
there is a valid traversal of the graph area and the frog can complete
the board using the cleanup gadget. If the board has a valid traversal, 
as seen above the sequences of nodes traversed in the graph area
corresponds to a Hamiltonian path in $G$ from $s$ to $t$.
The instance of the CFP can be constructed in polynomial time
because the size of the whole board is $3w+2v \times 7o_{n}w$,
where $w = 2^k-1, v = 2^{k-1}, o_n = 2n-1$, and
$2^k-1 < 8n$; so the board can be constructed in time $O(n^2)$.


The problem is in $\sf{NP}$ because a solution can easily be checked in polynomial time.
\end{proof}

\section{One dimensional variant}
\label{sec:1d}


It is easy to see that even if we restrict the board to be one dimensional
the problem remains $\sf{NP}$--complete.

\begin{theorem}The Crazy Frog Puzzle remains $\sf{NP}$--complete
even if restricted to 1-Dimensional boards ({\sc 1-D Crazy Frog Puzzle}),
i.e. boards of size $w \times 1$.
\end{theorem}

\begin{proof}
The immediate reduction is from the {\sc Crazy Frog Puzzle}: given an instance
of the CFP, i.e. a $n \times n$ partially filled board and a sequence of jumps: $\Delta_1,\Delta_2,\dots,\Delta_m$ check the sequence and if there is a jump $(\Delta x_i,\Delta y_i)$ such that 
$dx_i \geq n$ or $dy_i \geq n$ reject (the jump brings the frog outside of the board).
Otherwise expand it to size $3n \times 3n$ adding a border of blocked cells of width $n$
in all the four directions.
Then build an equivalent one dimensional crazy frog puzzle of size $(3n)^2$ putting the lines of the expanded board side by side (cells $(x,y)$ is mapped to cell $x+3ny$)
and converting every bidimensional jump $(\Delta x_i, \Delta y_i)$ to the one dimensional jump: $\Delta x_i + 3n*\Delta y_i$.
By construction every one dimensional $\Delta x_i + 3n*\Delta y_i$ jump will lead the frog from $x_0 + 3ny_0$ to the cell $(x_0 \pm \Delta x_i) + 3n *(y_0 \pm \Delta y)$ that corresponds to the original bidimensional cell $(x_0 \pm \Delta x_i, y_0 \pm \Delta y_i)$. Borders prevent the frog
to make moves that are invalid in the corresponding bidimensional configuration;
for example using a left jump $x+3n*0$ from cell $0 + 3n*1$ (on the left border in the bidimensional board) to reach another
part of the one dimensional board.
\end{proof}

Figure~\ref{fig:1d2} shows a simple example of a $3 \times 3$ CFP
instance transformed to a 1-Dimensional CFP board of length $81$.
\begin{figure}[htb]
\centering
\includegraphics[width=14.5cm]{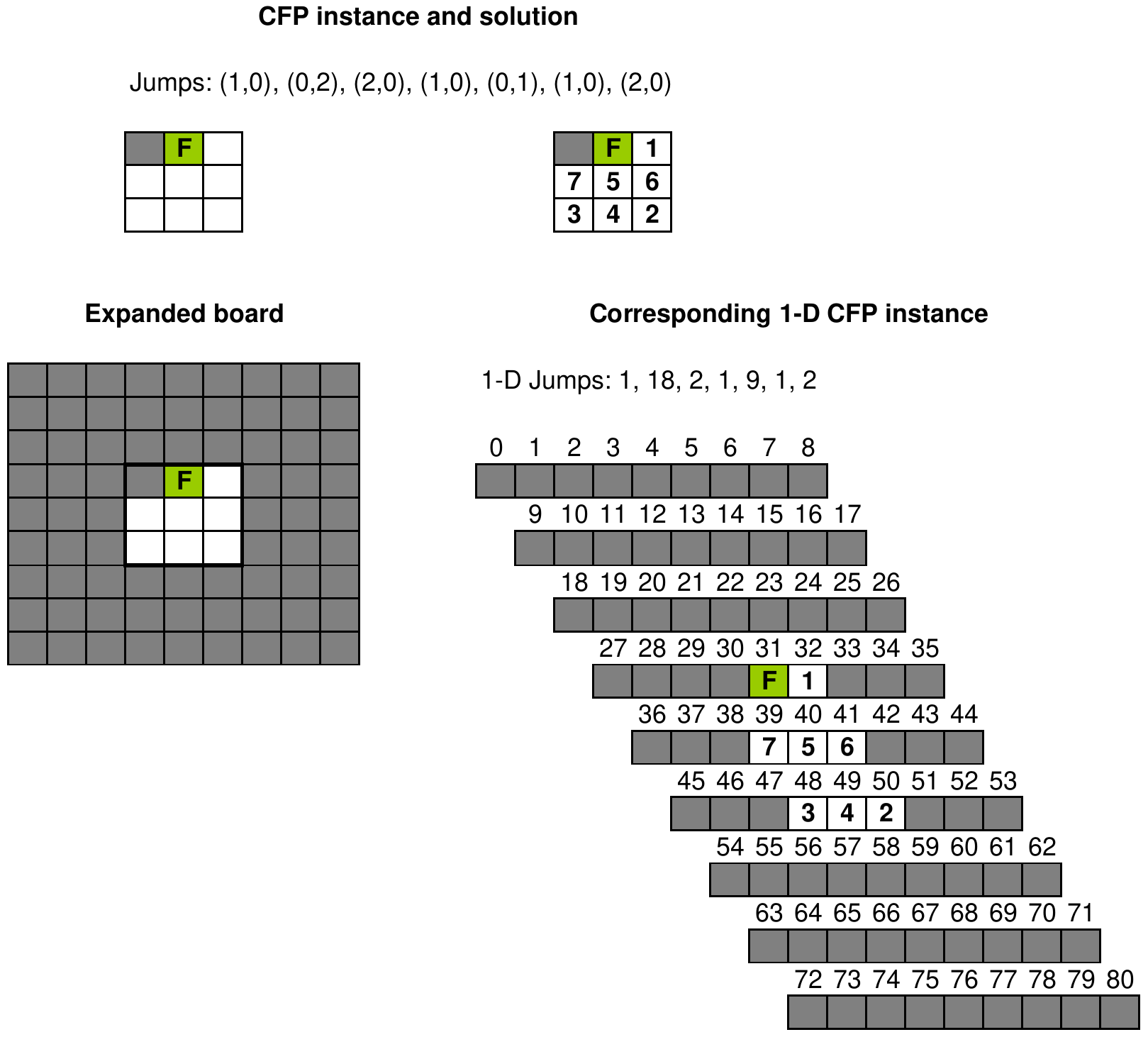}
\caption{A simple example of a $3 \times 3$ CFP
instance transformed to a 1-Dimensional CFP board of length $81$.}\label{fig:1d2}
\end{figure}

We can also fix the starting position of the frog:

\begin{lemma} Without loss of generality we can assume that
in the 1-D CFP instance the frog is placed on the 
leftmost cell.
\end{lemma}
\begin{proof}
Suppose that the 1-D board is $R = (E|B)^aF(E|B)^b$ and the sequence
of jumps is: $J_1,J_2,\dots,J_m$.
We can extend the board with a cell on the left that
will be the new starting position of the frog: 
$S' = FR$ and add a jump to the sequence: $(a+1,0),J_1,J_2,\dots,J_m$.
The first jump, that must be towards the right, places the frog on the original position. 
\end{proof}

\section{Permutation reconstruction from differences}
\label{sec:permutation}

We first prove that 1-D CFP is hard even if the initial board is empty.

\begin{lemma}1-D CFP remains $\sf{NP}$-complete even if the initial board is
empty.
\end{lemma}
\begin{proof}
Given an instance of the 1-D CFP, i.e. a configuration  $R = F\{B,E\}^{n-1}$ 
and a sequence $\tilde{J} = J_1,J_2,\dots,J_m$ of $m$ jumps;
suppose that $R$ contains $p = n - m - 1$ blocked cells at coordinates $x_1,x_2,\dots,x_p$;
let $d_1=x_1$, $d_i = x_i - x_{i-1}, i=2,3,\dots,p$, $d_{p+1}=n-x_p$.
We start with an empty line of length $2n+1$: $R' = FE^nE^n$
and extend the jump sequence in this way:
$$ \tilde{J}' = d_1+1,d_2,\dots,d_p,d_{p+1},\overbrace{1,1,\dots,1}^{n-1 \mbox{ times}}, 2n-1,\tilde{J}$$
(note that $|\tilde{J}'| = 2n$).
The $n-1$ steps forces a sequence of $n$ contiguous visited cells,
and it must be aligned with the rightmost part of the board
otherwise the frog will never be able to reach that cell during jumps $J_i$,
because $J_i < n$ (otherwise the original instance doesn't have a solution).
But, by construction, the only way to align it to the right is to make the
$d_i$ jumps towards the right, and they recreate exactly the $p$ blocked cells
of $R$. The jump $2n-1$ forces the frog to the second cell, which is also
the starting cell of the original configuration $R$.
The modified instance with the empty board has a solution if and only 
if the original instance has a solution.
\end{proof}
Figure~\ref{fig:empty} shows a 
reduction from 1-D CFP
to 1-D CFP with initial empty board.
\begin{figure}[htb]
\centering
\includegraphics[width=9cm]{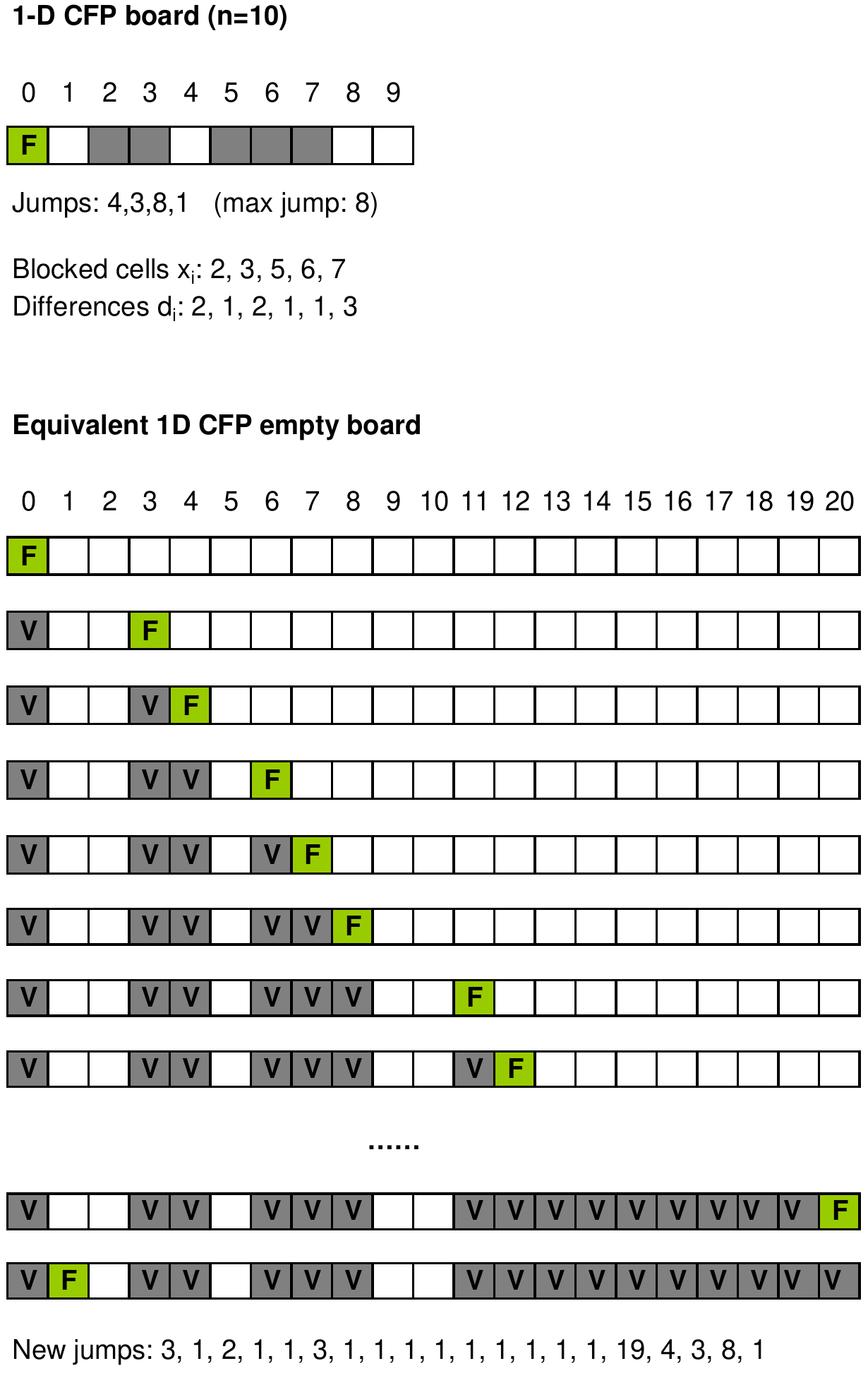}
\caption{An example of reduction from 1-D CFP
to 1-D CFP with initial empty board.}\label{fig:empty}
\end{figure}

\begin{definition}[{\sc Permutation Reconstruction from Differences}]~\\~\\
{\bf Input:} a set of $n-1$ distances $a_1, a_2,\dots,a_{n-1}$ with $a_i>0$\\
{\bf Question:} does there exist a permutation $\pi_1,\dots,\pi_n$ of the integers $[1..n]$
such that $| \pi_{i+1} - \pi_i| = a_i$, $i=1,\dots,n-1$?
\end{definition}
\smallskip
Note that if $\pi_1,\dots,\pi_n$ is a valid solution,
then the \emph{mirrored} sequence $n - \pi_1 + 1, n - \pi_2 + 1, \dots, n - \pi_n + 1$
is also a valid solution.
Figure~\ref{fig:prd} shows an example of a permutation reconstruction from differences problem.
The reduction from the 1-D CFP with initial empty board to the permutation reconstruction from differences problem is straightforward.

\begin{figure}[htp]
\centering
\includegraphics[width=7cm]{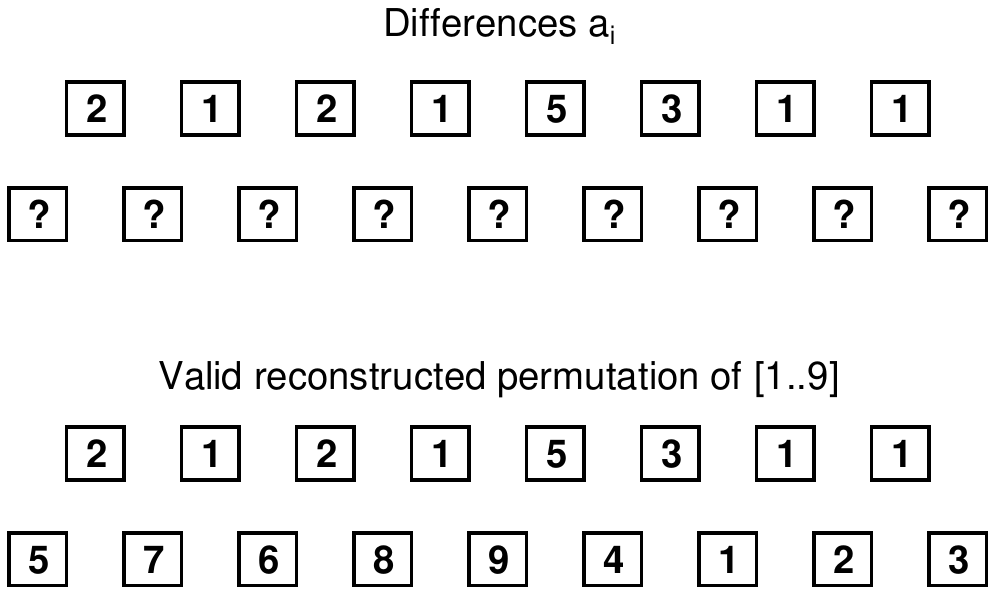}
\caption{An instance of the Permutation Reconstruction from Differences problem: the differences $a_i$ are $(2,1,2,1,5,3,1,1), n = 9$
and the reconstructed permutation is $(5,7,6,8,9,4,1,2,3)$;
the mirrored valid permutation is $(5,3,4,2,1,6,9,8,7)$ .}\label{fig:prd}
\end{figure}

\begin{theorem}
{\sc Permutation Reconstruction from Differences} (PRD) is $\sf{NP}$--complete.
\end{theorem}
\begin{proof}
Given an instance of the 1-D CFP with initial empty board of length $n$ and jumps $J_1,J_2,\dots,J_{n-1}$
it has a solution if and only if a valid permutation of $[1..n+1]$ can be reconstructed
from differences $a_1 = n$ and $a_{i} = J_{i-1}, i = 2,\dots,n$.

($\Rightarrow$) The frog visits all the cells of the board exactly once,
so its positions $x_i,\; i=1,\dots,n$ 
during the traversal (where $x_1=0$ is its starting position)
is a permutation of $[0..n-1]$ and it can be transformed
to a permutation of $[1..{n+1}]$ 
setting $\pi_{1}=n+1$ and $\pi_i = x_{i-1}+1,\; i=2,\dots,n+1$ .

($\Leftarrow$) Suppose that  $\pi_1,\dots,\pi_{n+1}$  is a valid permutation
that satisfy the difference constraints;
we have that $\pi_1 $ must be $1$ or $n+1$ because the first difference
$a_1 = n$. Suppose that $\pi_1=n+1$, then $\pi_2 = 1$ and $\pi_2 - 1,\ldots,\pi_{n+1}-1$ are a valid sequence of positions of the frog because
$|(\pi_{3}-1) - (\pi_{2}-1)| = J_1, |(\pi_{4}-1) - (\pi_{3}-1)| = J_2,\ldots$,
and they represent a valid solution to the 1-D CFP instance, too: 
the sign of jump $J_i$ is positive if $\pi_{i+2}>\pi_{i+1}$, negative
otherwise. If $\pi_1 = 1$ we can simply mirror
the values replacing every $\pi_i$ with $\pi'_i = (n+1) - \pi_i +1$ because
their absolute differences don't change. 

\end{proof}

\section{Conclusion}

We proved the hardness of a simple problem on permutations
that could shed light on other combinatorial or
arithmetic open problems.
For example there could be a correlation with the graceful
labeling problem,  indeed if the $a_i$ are themselves
a permutation of $[1..n]$ (all values are distinct) then the
permutation reconstruction from differences (PRD) problem
is equivalent to verify that the sequence is a graceful
labeling of the line of $n+1$.
So it would be interesting to study some restricted versions of 
the PRD problem; for example what is its complexity if the differences
are from a finite set of size $k$.  
As an intermediate step we introduced a new addictive
puzzle game that we hope will be soon playable online or
as a smartphone application.


\end{document}